\documentclass[conference]{IEEEtran}
\IEEEoverridecommandlockouts
\usepackage{eurosym}
\usepackage{amssymb}
\usepackage{amsmath}
\usepackage{amsfonts}
\usepackage{float}
\usepackage{graphics}
\usepackage[tight,footnotesize]{subfigure}
\usepackage{rotating}
\usepackage{algorithm}
\usepackage{enumerate}
\usepackage{cite}
\usepackage{xcolor}
\usepackage{color,soul}
\usepackage[colorinlistoftodos]{todonotes}
\usepackage{lipsum}
\usepackage{cases}
\usepackage{comment}
\usepackage{algpseudocode}
\usepackage{resizegather}

\usepackage{multirow}
\usepackage{threeparttable}

\usepackage{mathtools}

\DeclareMathOperator{\var}{Var}
\DeclareMathOperator{\diag}{diag}

\DeclareMathSizes{10}{9}{7}{5}
\newtheorem{theo}{Theorem}

\newtheorem{lemma}{Lemma}

\newtheorem{remark}{Remark}

\newtheorem{example}{Example}

\newcommand{\bs}{\boldsymbol{\mathrm{s}}}
\newcommand{\br}{\boldsymbol{\mathrm{r}}}
\newcommand{\be}{\boldsymbol{\mathrm{e}}}

\newcommand{\bu}{\boldsymbol{\mathrm{u}}}
\newcommand{\bt}{\boldsymbol{\mathrm{t}}}
\newcommand{\bA}{\boldsymbol{\mathrm{A}}}
\newcommand{\buone}{\bA^{T}\br+\be_{1}}

\newcommand{\coloneqq}{:=}

\begin{document}
\title{CRYSTALS-Kyber With Lattice Quantizer}
\specialpapernotice{}
\author{\IEEEauthorblockN{Shuiyin  Liu}
\IEEEauthorblockA{{Cyber Security Research and Innovation Centre} \\
{Holmes
Institute}\\
Melbourne, VIC 3000, Australia \\
Email: SLiu@Holmes.edu.au}
\and
\IEEEauthorblockN{Amin Sakzad}
\IEEEauthorblockA{{Department of Software Systems $\&$ Cybersecurity} \\
{Faculty of Information Technology, Monash University}\\
Melbourne, VIC 3800, Australia \\
Email: Amin.Sakzad@monash.edu}
}

\maketitle
\begin{abstract}
Module Learning with Errors (M-LWE) based key reconciliation mechanisms (KRM) can be viewed as quantizing an M-LWE sample according to a lattice codebook. This paper describes a generic M-LWE-based KRM framework, valid for any dimensional lattices and any modulus $q$ without a dither. Our main result is an explicit upper bound on the decryption failure rate (DFR) of M-LWE-based KRM. This bound allows us to construct optimal lattice quantizers to reduce the DFR and communication cost simultaneously. Moreover, we present a KRM scheme using the same security parameters $(q,k,\eta_1, \eta_2)$ as in Kyber. Compared with Kyber, the communication cost is reduced by up to $36.47\%$ and the DFR is reduced by a factor of up to $2^{99}$. The security arguments remain the same as Kyber. 
\end{abstract}

\section{Introduction}
Recent years have seen extraordinary progress in post-quantum cryptography, due to the fact that the public-key cryptosystems in use today are
known to be efficiently breakable by quantum computers. In August 2023, the National Institute of Standards and Technology (NIST) published draft post-quantum cryptography standards, where CRYSTALS-Kyber is selected for general encryption and key encapsulation mechanisms (KEM) \cite{NISTpqcdraft2023}. The security of Kyber relies on the Module Learning with Errors (M-LWE) assumption, which is widely believed to be post-quantum secure. The LWE family, including standard LWE \cite{Regev05}, Polynomial/Ring-LWE (R-LWE) \cite{RLWE2010}, Middle Product-LWE (MP-LWE) \cite{RSSS17}, and M-LWE \cite{BS+15}, has been used to construct many types of advanced cryptosystems, such as fully homomorphic encryption \cite{homomorphiLWE2011} and digital signature \cite{NISTpqcdraftDS2023}.  However, LWE based encryption schemes generally require a large modulus $q$, resulting in large communication overhead. For example, Kyber introduces a $24-49$ times more communication overhead compared to the Elliptic Curve-based public-key cryptosystem.
\begin{sloppypar}
The communication overhead can be measured by the ciphertext expansion rate (CER), i.e., the ratio of the ciphertext size to the plaintext size. Recent studies on LWE-based cryptosystems focus on reducing the CER \cite{CiphertextQuantization2023} \cite{NewhopeECC2018} \cite{FrodoCong2022}\cite{liu2023lattice}. In \cite{CiphertextQuantization2023}, the asymptotic value of CER was derived for LWE-based encryption approaches.  In \cite{NewhopeECC2018}, a concatenation of BCH and LDPC codes was proposed for NewHope-Simple, to reduce the CER by $12.8\%$. In \cite{FrodoCong2022}, lattice codes were introduced to FrodoKEM, to reduce the CER by a factor of up to $6.7\%$. In \cite{liu2023lattice}, a coded modulation scheme using lattices and BCH codes was applied to Kyber, where the CER is decreased by a factor of $24.49\%$.
In summary, the idea behind the above approaches is to reduce the decryption failure rate (DFR) via error-correcting techniques: A smaller DFR allows the use of a smaller modulus $q$ \cite{FrodoCong2022} or further compression of the ciphertext size \cite{NewhopeECC2018}\cite{liu2023lattice}. Hence the CER is decreased.
\end{sloppypar}
Since LWE-based encryption schemes commonly compress, or more generally, quantize
ciphertexts to reduce CER \cite{NISTpqcdraft2023}\cite{NewhopeECC2018}, an interesting idea is to design a key exchange protocol purely based on quantizing an LWE sample. In the literature, this approach is referred to as key reconciliation mechanisms (KRM): Given quantized LWE samples, Alice and Bob can compute the same shared secret.
A standard LWE-based KRM scheme was proposed in \cite{KRMPeikert2014}, where a $1-$dimensional rounding-off quantizer is employed. In \cite{NewHope2016}, $\mathsf{D_4}$ lattice quantizer was used to construct an R-LWE based KRM scheme. This approach was further extended to the M-LWE scenario with $\mathsf{E_8}$ lattice quantizer and an even modulus $q$ in \cite{MLWEE82021}. Intuitively speaking, KRM-based approaches should enjoy a smaller CER compared with Kyber, since they use better quantizers. However, the CERs in \cite{NewHope2016} and \cite{MLWEE82021} are larger than KYBER768 \cite{Kyber2021}, by a factor of $88.2\%$ and $8.8\%$, respectively. Another issue is the choice of $q$. To ensure the shared secret is uniformly random, KRM schemes require either an even $q$ \cite{MLWEE82021} or a prime $q$ with a dither \cite{KRMPeikert2014}\cite{NewHope2016}. Both choices will increase the computation complexity: the former prevents using the Number Theoretic Transform (NTT) for efficient polynomial multiplication (Karatsuba is slower than NTT), while the latter adds more computing steps. An open question is how to construct a KRM scheme with a smaller CER than Kyber. 

In this work, we propose a generic M-LWE-based KRM framework, valid for any dimensional lattices and any modulus $q$ without a dither.
Note that the work in \cite{MLWEE82021} only considers $\mathsf{E_8}$ lattice, and moreover does not include explicit bounds on DFR. Thus it cannot be generalized to higher-dimensional lattices.
We first remove the constraint of even $q$ in \cite{MLWEE82021}.
The idea is rejection sampling: for a prime $q$, we accept a M-LWE sample from the ring $R_{q-1}$ rather than $R_{q}$. The rejection probability is about $n/q$, where $n$ is the degree of the ring $R_{q}$. Secondly, we design a generic lattice quantizer using an integer generator matrix, which is valid for any dimension. Thirdly, we derive an explicit upper bound on the DFR of the proposed M-LWE-based KRM framework. This bound tells us to select lattices with large packing radius and small covering radius, since they can reduce the DFR and CER simultaneously. We consider $16-$dimensional Barnes–Wall (BW16) and $24-$dimensional Leech lattices (Leech24), since they outperform $\mathsf{D_4}$ and $\mathsf{E_8}$ lattices used in \cite{NewHope2016}\cite{MLWEE82021}. Last but not least, we demonstrate a M-LWE-based KRM scheme using the same security parameters $(q,k,\eta_1, \eta_2)$ as in KYBER768. Compared with KYBER768, the CER is reduced by up to $36.47\%$ and the DFR is reduced by a factor of up to $2^{99}$. We summarize our results in Table \ref{sum_contribution}. The security arguments remain the same as KYBER768.
\begin{table}[th]
\centering
\begin{threeparttable}[b]
\caption{Comparison of KYBER768 and KRMs}
\label{sum_contribution}\centering
\vspace{-3mm} 
\begin{tabular}{|c|c|c|c|c|}
\hline
& $q$ & CER  &  DFR & Source \\ \hline
KYBER768 & $3329$ & $34$ & $2^{-164}$ &\cite{NISTpqcdraft2023}\cite{Kyber2021}\\ \hline
KRM-$\mathsf{E_8}$ & $2048$ & $37$ & $2^{-174}$ &  \cite{MLWEE82021}\\ \hline
KRM-$\mathsf{E_8}$ & $3329$ & $31$ & $2^{-174}$ &  This work\\ \hline
KRM-$\mathsf{BW16}$ & $3329$ & $26.4$ & $2^{-263}$ & This work \\ \hline
KRM-$\mathsf{Leech24}$ & $3329$ & $21.6$ & $2^{-172}$ & This 
work \\ \hline
\end{tabular}
\vspace{-1mm}
\end{threeparttable}
\end{table}
\vspace{-1mm}

\emph{Organization:} Section II presents notation, Kyber, and $\mathsf{E_8}$ based KRM scheme. Section III describes the proposed KRM framework. Section IV gives the DFR analysis of the proposed KRM framework. Section V shows an example of the proposed KRM scheme. Section VI concludes with conclusion remarks. 


\section{System Model}
\subsection{Notation}
\emph{Rings:} We use $R$ and $R_{q}$ to represent the rings $\mathbb{Z}[X]/(X^{n}+1)$ and $\mathbb{Z}_{q}[X]/(X^{n}+1)$, respectively. In this work, the value of $n$ is $256$.
Matrices and vectors are represented as bold upper-case and lower-case letters, respectively. The transpose of a vector $\mathbf{v}$ or a matrix $\mathbf{A}$ is denoted by $\mathbf{v}^T$ or $\mathbf{A}^T$, respectively.  Elements in $R$ or $R_{q}$ are denoted by regular font letters, while a vector of the coefficients in $R$ or $R_{q}$ is represented by bold lower-case letters. The default vectors are column vectors.


\emph{Sampling and Distribution:} For a set $\mathcal{S}$, we write $s \leftarrow \mathcal{S}$ to mean that $s$
is chosen uniformly at random from $\mathcal{S}$. If $\mathcal{S}$ is a probability
distribution, then this denotes that $s$ is chosen according to
the distribution $\mathcal{S}$. For a polynomial $f(x) \in R_q$ or a vector of such polynomials, this notation is defined coefficient-wise. We use $\var(\mathcal{S})$ to represent the variance of the distribution $\mathcal{S}$. Let $x$ be a bit string and $S$ be a distribution taking $x$ as the input, then $y\sim S\coloneqq\mathsf{Sam}\left(x\right)$ represents that the output $y$ generated by distribution $S$ and input $x$ can be extended to any desired length. We denote $\beta_{\eta}=B(2\eta
,0.5)-\eta $ as the central binomial distribution over $\mathbb{Z}$.




\emph{Lattice, Product Lattice, and Lattice Quantizer:}
An $\ell$-dimensional lattice $\Lambda$ is a
discrete additive subgroup of $\mathbb{R}^m$, $m \leq \ell$. Based on $\ell$ linearly independent vectors $b_1, \ldots, b_\ell$ in $\mathbb{R}^m$, $\Lambda$ can be
written as
$\Lambda =\mathcal{L}(\mathbf{B})=z_{1}\mathbf{b}_{1}+\cdots +z_{\ell}\mathbf{b}%
_{\ell}$, 
where $z_{1},\ldots,z_{\ell}\in \mathbb{Z}$, and $\mathbf{B}=[\mathbf{b}_{1},\ldots,%
\mathbf{b}_{\ell}]$ is referred to as a generator matrix of $\Lambda$. The volume (or determinant) of a lattice $\Lambda$ is $\mathsf{Vol}(\Lambda)=|\mathsf{det}(\mathbf{B})|$. For any $\mathbf{x}\in \mathbb{R}^{\ell}$, the closest vector in $\Lambda$ to $\mathbf{x}$ is denoted as $\mathsf{Q}_{\Lambda}(\mathbf{x})$. The \emph{Voronoi region} of $\Lambda$, denoted by $\mathcal{V}(\Lambda)$, is the set of all points in $\mathbb{R}^{\ell}$ which are closest to the origin than to any other lattice
point.  The \emph{fundamental Voronoi region} of $\Lambda$, denoted by $\mathcal{V}_0(\Lambda)$, is the set of points that are closer to $0$ lattice point than to any other lattice point. The modulo $\Lambda$ operation is denoted as $\mathbf{x} \bmod \Lambda = \mathbf{x} - \mathsf{Q}_{\Lambda}(\mathbf{x})$. We consider product lattices, which are the Cartesian
product of two or more lower-dimensional lattices $\Lambda=\Lambda_1 \times \cdots \times \Lambda_t$. If $\Lambda_1 = \cdots = \Lambda_t$, we write $\Lambda=\Lambda_1^t$. Given the lattices $ \Lambda_3  \subseteq \Lambda_2  \subseteq \Lambda_1$ and $\mathbf{y}\in \mathbb{R}^{\ell}$, we define the lattice quantizers by \cite{MLWEE82021}:
\begin{align}
\mathsf{HelpRec}(\mathbf{x})&=\mathsf{Q}_{\Lambda_1}(\mathbf{x}) \bmod \Lambda_2,\nonumber\\
\mathsf{Rec}(\mathbf{x},\mathbf{y})&=\mathsf{Q}_{\Lambda_2}(\mathbf{x}-\mathbf{y}) \bmod \Lambda_3.\label{LatticeQuantizer}
\end{align}

\emph{Compress and Decompress ($1$-Dimensional Lattice Quantizer):} Let $x\in\mathbb{R}$ be a real number, then $\left\lceil x\right\rfloor $ means rounding to the closet integer with ties rounded up. The operations $\left\lfloor x\right\rfloor $ and $\left\lceil x \right\rceil$ mean rounding $x$ down and up, respectively. Let $x \in \mathbb{Z}_{q}$ and $d \in \mathbb{Z}$ be such that $2^d<q$. We define compression and decompression functions by \cite{NISTpqcdraft2023}:
{\setlength{\abovedisplayskip}{3pt} \setlength{\belowdisplayskip}{3pt}
\begin{align}
\mathsf{Compress}_{q}(x,d)&=\lceil (2^{d}/q)\cdot x\rfloor \mod 2^{d},\nonumber\\
\mathsf{Decompress}_{q}(x,d)&=\lceil (q/2^{d})\cdot x\rfloor.\label{ComDecom}
\end{align}}


\emph{Decryption Failure Rate (DFR) and Ciphertext Expansion Rate (CER):}
We let $\text{DFR} =\delta := \Pr (\hat{m} \neq m)$, where $m$ is a shared secret. It is desirable to have a small $\delta$, in order to be safe against decryption failure attacks \cite{DFRAttack2019}. In this work, the communication cost refers to the ciphertext expansion rate (CER), i.e., the ratio of the ciphertext size to the plaintext size.

\subsection{Kyber KEM}
Let $\mathcal{M}_{2,n} = \{0,1\}^{n}$ denote the
message space, where every message $m \in \mathcal{M}_{2,n}$ can be
viewed as a polynomial in $R$ with coefficients in $\{0,1\}$. Consider the public-key encryption scheme Kyber.CPA = (KeyGen; Enc; Dec) as described in Algorithms 2.2.1 to 2.2.3 \cite{Kyber2021}. The values of $\delta$, CER, and $(q, k, \eta_1,\eta_2, d_u, d_v)$ are given in Table \ref{Kyber_Par}. Note that the parameters $(q, k, \eta_1,\eta_2)$ determine the security level of Kyber, while $(d_u, d_v)$ decide ciphertext compression rate, or equivalently, the communication overhead.
\vspace{-3mm}
\begin{algorithm}[H]
\caption{$\mathsf{Kyber.CPA.KeyGen()}$: key generation}
\label{alg:kyber_keygen}
\begin{algorithmic}[1]

    \State
    $\rho,\sigma\leftarrow\left\{ 0,1\right\} ^{256}$

    \State
    $\bA\sim R_{q}^{k\times k}\coloneqq\mathsf{Sam}(\rho)$

    \State
    $(\bs,\be)\sim\beta_{\eta_1}^{k}\times\beta_{\eta_1}^{k}\coloneqq\mathsf{Sam}(\sigma)$

    \State
    $\bt\coloneqq\boldsymbol{\mathrm{As+e}}$\label{line:t}

    \State \Return $\left(pk\coloneqq(\boldsymbol{\mathrm{t}},\rho),sk\coloneqq\bs\right)$  

\end{algorithmic}
\end{algorithm}

\vspace{-5mm}

\begin{algorithm}[H]
\caption{$\mathsf{Kyber.CPA.Enc}$ $(pk=(\boldsymbol{\mathrm{t}},\rho),m\in\mathcal{M}_{2,n})$}
\label{alg:kyber_enc}
\begin{algorithmic}[1]

	\State
	$r \leftarrow \{0,1\}^{256}$
	
	
	\State
	$\boldsymbol{\mathrm{A}}\sim R_{q}^{k\times k}\coloneqq\mathsf{Sam}(\rho)$
	
	\State  $(\boldsymbol{\mathrm{r}},\boldsymbol{\mathrm{e}_{1}},e_{2})\sim\beta_{\eta_1}^{k}\times\beta_{\eta_2}^{k}\times\beta_{\eta_2}\coloneqq\mathsf{Sam}(r)$
	
	\State  $\boldsymbol{\mathrm{u}}\coloneqq\mathsf{Compress}_{q}(\buone,d_{u})$\label{line:u}
	
	\State  $v\coloneqq\mathsf{Compress}_{q}(\boldsymbol{\mathrm{t}}^{T}\boldsymbol{\mathrm{r}}+e_2+\left\lceil {q}/{2}\right\rfloor \cdot m,d_{v})$\label{line:v}
	
	\State \Return $c\coloneqq(\boldsymbol{\mathrm{u}},v)$

\end{algorithmic}
\end{algorithm}

\vspace{-5mm}

\begin{algorithm}[H]
\caption{${\mathsf{Kyber.CPA.Dec}}\ensuremath{(sk=\bs,c=(\bu,v))}$}
\begin{algorithmic}[1]

    \State
    $\bu\coloneqq\mathsf{Decompress}_{q}(\bu,d_{u})$

    \State
    $v\coloneqq\mathsf{Decompress}_{q}(v,d_{v})$

    \State \Return $\mathsf{Compress}_{q}(v-\bs^{T}\bu,1)$

\end{algorithmic}
\end{algorithm}

\vspace{-5mm}

\begin{table}[ht]
\caption{Parameters: Kyber in \cite{NISTpqcdraft2023}\cite{Kyber2021} vs. KRM-$\mathsf{E}_8$ in \cite{MLWEE82021}}
\label{Kyber_Par}\centering
\vspace{-3mm}
\begin{tabular}{|c|c|c|c|c|c|c|c|c|c|}
\hline
& $k$ & $q$ & $\eta_{1}$ & $\eta_{2}$ & $d_{u}$ & $d_{v}$ & $\delta$ & CER\\ \hline
KYBER512 &  $2$ & $3329$ & $3$ & $2$ & $10$ & $4$ & $2^{-139}$ & $24$\\ \hline
KYBER768 & $3$ & $3329$ & $2$ & $2$ & $10$ & $4$ & $2^{-164}$ &$34$\\ \hline
KYBER1024 & $4$ & $3329$ & $2$ & $2$ & $11$ & $5$ & $2^{-174}$ &$49$ \\ \hline
KRM-$\mathsf{E}_8$ & $3$ & $2048$ & $2$ & $2$ & $11$ & $4$ & $2^{-174}$ &$37$ \\ \hline
\end{tabular}
\vspace{-2mm}
\end{table}

\begin{table}[ht]
\caption{Security: Kyber in \cite{NISTpqcdraft2023}\cite{Kyber2021} vs. KRM-$\mathsf{E}_8$ in \cite{MLWEE82021}}
\label{Kyber_Security}\centering
\vspace{-3mm} 
\begin{tabular}{|c|c|c|c|c|}
\hline
& KYBER512 & KYBER768  & KYBER1024 & KRM-$\mathsf{E}_8$\\ \hline
NIST level & $1$ & $3$ & $5$ & $3$\\ \hline
Classical &$118$& $183$ & $256$  & $194$\\ \hline
Quantum & $107$ & $166$ & $232$ &  $176$\\ \hline
\end{tabular}
\vspace{-5mm}
\end{table}

\subsection{M-LWE-based KRM with $\mathsf{E_8}$ Lattice (KRM-$\mathsf{E_8}$)}
The KRM based key agreement approach in \cite{MLWEE82021} uses $\mathsf{E_8}$ lattice based quantizer. For an integer $p>1$, we choose
\begin{equation}
\setlength{\abovedisplayskip}{3pt}
\setlength{\belowdisplayskip}{3pt}
\Lambda_1=(q/2^{p} \cdot   \mathsf{E_8})^{32}, \ \Lambda_2 =2^{p-1}\Lambda_1, \ \Lambda_3 = q (\mathbb{Z}^8)^{32},\label{NestE8}
\end{equation}
Considering the  KRM-$\mathsf{E_8}$.CPA =
(KeyGen; Enc; Dec) as described in Algorithms 2.3.1 to 2.3.3 \cite{MLWEE82021}. 

\setcounter{algorithm}{0}

\vspace{-3mm}
\begin{algorithm}[H]
\caption{$\mathsf{KRM-{E}_8.CPA.KeyGen()}$: key generation}
\label{alg:KRM_keygen}
\begin{algorithmic}[1]

    \State
    $\rho,\sigma\leftarrow\left\{ 0,1\right\} ^{256}$

    \State
    $\bA\sim R_{q}^{k\times k}\coloneqq\mathsf{Sam}(\rho)$

    \State
    $(\bs,\be)\sim\beta_{\eta_1}^{k}\times\beta_{\eta_1}^{k}\coloneqq\mathsf{Sam}(\sigma)$

    \State
    $\bt\coloneqq\boldsymbol{\mathrm{As+e}}$\label{line:t2}

    \State \Return $\left(pk\coloneqq(\boldsymbol{\mathrm{t}},\rho),sk\coloneqq\bs\right)$  

\end{algorithmic}
\end{algorithm}

\vspace{-5mm}

\begin{algorithm}[H]
\caption{$\mathsf{KRM-{E}_8.CPA.Enc}$ $(pk=(\boldsymbol{\mathrm{t}},\rho)$}
\label{alg:KRM_enc}
\begin{algorithmic}[1]

	\State
	$r \leftarrow \{0,1\}^{256}$
	
	
	\State
	$\boldsymbol{\mathrm{A}}\sim R_{q}^{k\times k}\coloneqq\mathsf{Sam}(\rho)$
	
	\State  $(\boldsymbol{\mathrm{r}},\boldsymbol{\mathrm{e}_{1}},e_{2})\sim\beta_{\eta_1}^{k}\times\beta_{\eta_2}^{k}\times\beta_{\eta_2}\coloneqq\mathsf{Sam}(r)$
	
	\State  $\boldsymbol{\mathrm{u}}\coloneqq\buone$\label{line:u2}
	
	\State  $v\coloneqq\mathsf{HelpRec}(\boldsymbol{\mathrm{t}}^{T}\boldsymbol{\mathrm{r}}+e_2)$\label{line:v2}

 \State $m\coloneqq \mathsf{Rec}( \boldsymbol{\mathrm{t}}^{T}\boldsymbol{\mathrm{r}}+e_2,v)$
	
	\State \Return $c\coloneqq(\boldsymbol{\mathrm{u}},v)$

\end{algorithmic}
\end{algorithm}

\vspace{-5mm}

\begin{algorithm}[H]
\caption{${\mathsf{KRM-{E}_8.CPA.Dec}}\ensuremath{(sk=\bs,c=(\bu,v))}$}
\begin{algorithmic}[1]

    \State \Return $m=\mathsf{Rec}(\bs^{T}\bu,v)$

\end{algorithmic}
\end{algorithm}

\vspace{-3mm}

Comparing Algorithms 2.2.1 - 2.2.3 with Algorithms 2.3.1 - 2.3.3, we see that the key generation functions are the same. The major difference lies in encryption: Kyber adds a secret $m$ to the M-LWE sample $\boldsymbol{\mathrm{t}}^{T}\boldsymbol{\mathrm{r}}+e_2$, while KRM-$\mathsf{E_8}$ computes a common secret $m$ by quantizing the M-LWE sample $\boldsymbol{\mathrm{t}}^{T}\boldsymbol{\mathrm{r}}+e_2$. 

The parameters of KRM-$\mathsf{E_8}$ are listed in Table \ref{Kyber_Par}. Using the same notation as in Kyber \cite{Kyber2021}, we have $(d_u=\log_2(q), d_v=p-1)$. Since KRM-$\mathsf{E_8}$ requires an even $q$, the Karatsuba algorithm is used to speed up polynomial multiplications. But it is slower than the NTT used in Kyber.

Table \ref{Kyber_Security} shows the security levels of KRM-$\mathsf{E_8}$ and Kyber. We see that KRM-$\mathsf{E_8}$ and KYBER768 are comparable. However, as shown in Table \ref{Kyber_Par}, the CER of KRM-$\mathsf{E_8}$ is slightly higher than KYBER768. One reason is that Kyber compresses both parts of the ciphertext $c$, i.e.,
\begin{equation}
 (\mathsf{Compress}_{q}(\buone,d_{u}),\mathsf{Compress}_{q}(\boldsymbol{\mathrm{t}}^{T}\boldsymbol{\mathrm{r}}+e_2+\left\lceil {q}/{2}\right\rfloor \cdot m,d_{v})
\end{equation}
 while KRM-$\mathsf{E_8}$ only quantizes the second part of $c$, i.e.,
\begin{align}
(\buone,\mathsf{HelpRec}(\boldsymbol{\mathrm{t}}^{T}\boldsymbol{\mathrm{r}}+e_2)).
\end{align}

Another issue is that an upper bound on DFR is evaluated numerically in \cite{MLWEE82021}. No closed-form expression is provided. It is unclear how the choice of lattices affect the DFR/CER.

In this work, we will develop a generic KRM framework, valid for any dimensional lattices and any $q$ without a dither. Our design will quantize both parts in $c$, so that the value of CER is minimized. We will explore the trade-off between DFR and CER, for a given set of security parameters.

\vspace{-0mm}
\section{The Proposed KRM Framework } 
\vspace{-0mm}

\subsection{Generic Lattice Quantizer} 
For an arbitrary $q$ and a $\ell-$dimensional lattice $\Lambda$ with an integer generator matrix $\mathbf{B}$, we consider the following lattices
\begin{align}
\Lambda_1 =\lfloor q/2^{{p}}\rfloor \Lambda^{n/\ell}, \ \Lambda_2 = 2^{{p} -{t}}\Lambda_1, \  \Lambda_3 =2^{{p}} \lfloor q/2^{{p}} \rfloor \mathbb{Z}^{n},  \label{prime_quantizer}
\end{align}
where the integer $t<p$ is selected to ensure $ \Lambda_3  \subseteq \Lambda_2  \subseteq \Lambda_1$. It is clear that $\Lambda_2  \subseteq \Lambda_1$, we then show how to find an integer sub-lattice of $\Lambda_2$, or equivalently, an integer sub-lattice of $\Lambda$.

\begin{lemma}
Considering the Smith Normal Form factorization (SNF) of a $\ell-$dimensional lattice $\Lambda$ with an integer generator matrix $\mathbf{B}$, denoted as $\mathbf{B}= \mathbf{U}\cdot \diag (\pi_1, \ldots, \pi_\ell) \cdot \mathbf{U}'$, where $\mathbf{U}, \mathbf{U}' \in \mathbb{Z}^{\ell \times \ell}$ are unimodular matrices. Let $\pi_{\mathsf{LCM}}$ be the least common multiplier of $\pi_1, \ldots, \pi_\ell$. We have
\begin{equation}
\setlength{\abovedisplayskip}{3pt}
\setlength{\belowdisplayskip}{3pt}
\pi_{\mathsf{LCM}} \mathbb{Z}^{\ell} \subseteq \Lambda.
\end{equation}
\end{lemma}
\begin{proof}
Let $\mathbf{x} \in \mathbb{Z}^{\ell}$. We have $\pi_{\mathsf{LCM}}\mathbf{x} \in \pi_{\mathsf{LCM}}\mathbb{Z}^{\ell}$. To show  $\pi_{\mathsf{LCM}} \mathbb{Z}^{\ell} \subseteq \Lambda$, we need to prove $\mathbf{B}^{-1}\pi_{\mathsf{LCM}}\mathbf{x}$ is an integer vector. Since $\mathbf{B}^{-1}\pi_{\mathsf{LCM}} = \mathbf{U}'^{-1}\cdot \diag (\pi_\mathsf{LCM}/\pi_1, \ldots, \pi_\mathsf{LCM}/\pi_\ell) \cdot \mathbf{U}^{-1}$ is an integer matrix, the proof is completed.
\end{proof}

According to Lemma 1, the values of $t$ in (\ref{prime_quantizer}) is equal to
\begin{equation}
\setlength{\abovedisplayskip}{3pt}
\setlength{\belowdisplayskip}{3pt}
t=\log_2(\pi_{\mathsf{LCM}}). \label{value_t}
\end{equation}
To ensure $t$ is an integer, $\pi_{\mathsf{LCM}}$ is required to be a power of $2$. In other words, the determinant of $\Lambda$ is a power of $2$.
\begin{remark}
For the choice of $\Lambda$, we consider $\mathsf{E_8}$ lattice, Barnes–Wall lattice with $\ell=16$ (BW16), and Leech lattice with $\ell=24$ (Leech24)\cite{BK:Conway93}. They provide the optimal density in the corresponding dimensions and have a fixed-complexity quantizer \cite{FrodoCong2022}\cite{VD93Leech}. For BW16 and Leech24 lattices, we will scale the original generator matrix to an integer matrix $\mathbf{B} \in \mathbb{Z}^{\ell \times \ell}$. For $\mathsf{E_8}$ lattice, we consider the original basis as in \cite{MLWEE82021}, since $\lfloor q/2^{{p}}\rfloor$ is even in this work. The determinant of these lattices is a power of $2$.
\end{remark}

\begin{example}
One can easily calculate ${t}=1, 2, 3$ for $\Lambda= \mathsf{E_8}, \mathsf{BW16}, \mathsf{Leech24}$, respectively.
\end{example}

\subsection{Rejection Sampling} 
To ensure the security of a KRM based on (\ref{prime_quantizer}), we need to show the common secret $m$ in Algorithm 2.3.2 is uniformly distributed. Let $\hat{q}=2^{{p}} \lfloor q/2^{{p}} \rfloor$, we recall the below Lemma.

\begin{lemma}[\cite{MLWEE82021}]
If $x \in R_{\hat{q}}$ is uniformly random, then $m =
\mathsf{Rec}(x; v)$ is uniformly random, given $v = \mathsf{HelpRec}(x)$.
\end{lemma}

According to Algorithm 2.3.2, we have $x=\boldsymbol{\mathrm{t}}^{T}\boldsymbol{\mathrm{r}}+e_2 \in R_q$, which is uniformly random according to the M-LWE assumption. For an arbitrary $q$, since $q \geq \hat{q}$, we will apply rejection sampling on $x$ and select those $x \in R_{\hat{q}}$, which is still uniformly random. The rejection probability, denoted as $P_\mathbf{rej}$, is given by
\begin{equation}
P_\mathsf{rej}=1-(\hat{q}/q)^n \approx n(1-\hat{q}/q).
\end{equation}
The approximation holds when $n|1-\hat{q}/q| \ll 1 $. Through rejection sampling, we remove the constraint of even $q$ in \cite{MLWEE82021}.
\begin{example}
Given $q=3329$ used in Kyber, we have $\hat{q}=3328$ for $1<p<9$ and $P_\mathsf{rej} \approx 7.69\%$.
\end{example}

More details to be followed in the next subsection.

\subsection{The Proposed KRM-$\Lambda$}

Based on the results in the above subsections, we give here the proposed generic KRM framework, denoted as  KRM-$\Lambda$. The KeyGen function is the same as in Algorithm 2.3.1. The proposed Enc/Dec functions are described in Algorithms 3.3.2 and 3.3.3. To reduce CER, we also compress $\mathbf{u}$ as in Kyber.

\setcounter{algorithm}{0}

\vspace{-3mm}
\begin{algorithm}[H]
\caption{$\mathsf{KRM-\Lambda.CPA.KeyGen()}$: key generation}
\label{alg:KRM_L_keygen}
\begin{algorithmic}[1]

    \State
    $\rho,\sigma\leftarrow\left\{ 0,1\right\} ^{256}$

    \State
    $\bA\sim R_{q}^{k\times k}\coloneqq\mathsf{Sam}(\rho)$

    \State
    $(\bs,\be)\sim\beta_{\eta_1}^{k}\times\beta_{\eta_1}^{k}\coloneqq\mathsf{Sam}(\sigma)$

    \State
    $\bt\coloneqq\boldsymbol{\mathrm{As+e}}$\label{line:t3}

    \State \Return $\left(pk\coloneqq(\boldsymbol{\mathrm{t}},\rho),sk\coloneqq\bs\right)$  

\end{algorithmic}
\end{algorithm}

\vspace{-5mm}

\begin{algorithm}[H]
\caption{$\mathsf{KRM-\Lambda.CPA.Enc}$ $(pk=(\boldsymbol{\mathrm{t}},\rho)$}
\label{alg:KRM_L_enc}
\begin{algorithmic}[1]

	\State
	$r \leftarrow \{0,1\}^{256}$
	
	
	\State
	$\boldsymbol{\mathrm{A}}\sim R_{q}^{k\times k}\coloneqq\mathsf{Sam}(\rho)$
 
\If{$q \neq \hat{q}$}   \hfill{//rejection sampling}

 \While {$\max(\boldsymbol{\mathrm{t}}^{T}\boldsymbol{\mathrm{r}}+e_2)>\hat{q}-1$}
	
	\State  $(\boldsymbol{\mathrm{r}},\boldsymbol{\mathrm{e}_{1}},e_{2})\sim\beta_{\eta_1}^{k}\times\beta_{\eta_2}^{k}\times\beta_{\eta_2}\coloneqq\mathsf{Sam}(r)$
	

 \EndWhile                         \hfill{//select $\boldsymbol{\mathrm{t}}^{T}\boldsymbol{\mathrm{r}}+e_2 \in R_{\hat{q}}$}
 \EndIf

 \State  $\boldsymbol{\mathrm{u}}\coloneqq\mathsf{Compress}_{q}(\buone,d_{u})$
 
	\State  $v\coloneqq\mathsf{HelpRec}(\boldsymbol{\mathrm{t}}^{T}\boldsymbol{\mathrm{r}}+e_2)$\label{line:v3}

 \State $m\coloneqq \mathsf{Rec}( \boldsymbol{\mathrm{t}}^{T}\boldsymbol{\mathrm{r}}+e_2,v)$
	
	\State \Return $c\coloneqq(\boldsymbol{\mathrm{u}},v)$

\end{algorithmic}
\end{algorithm}

\vspace{-5mm}

\begin{algorithm}[H]
\caption{${\mathsf{KRM-\Lambda.CPA.Dec}}\ensuremath{(sk=\bs,c=(\bu,v))}$}
\begin{algorithmic}[1]

    \State
    $\bu'\coloneqq\mathsf{Decompress}_{q}(\bu,d_{u})$

    \State \Return $m=\mathsf{Rec}(\bs^{T}\bu',v)$

\end{algorithmic}
\end{algorithm}

\vspace{-3mm}

\begin{remark}
The Proposed KRM-$\Lambda$ framework compresses both parts of the ciphertext $c$, i.e,
\begin{equation}
\setlength{\abovedisplayskip}{3pt}
\setlength{\belowdisplayskip}{3pt}
 (\mathsf{Compress}_{q}(\buone,d_{u}),\mathsf{HelpRec}(\boldsymbol{\mathrm{t}}^{T}\boldsymbol{\mathrm{r}}+e_2)). \label{c_Lambda}
\end{equation}
With $\Lambda = \mathsf{E}_8$, $q=2048$, and $d_u=11$, the proposed scheme reduces to KRM-$\mathsf{E}_8$ \cite{MLWEE82021}. By tuning the values of $(q, d_u)$, we can generate different variants of KRM-$\mathsf{E}_8$, which could have smaller CER than the original version. Furthermore, when $q$ is prime, NTT can be applied to speed up the polynomial multiplications in Algorithms 3.3.1-3.3.3. More importantly, the proposed KRM-$\Lambda$ framework enables the use of better lattices than $\mathsf{E}_8$, e.g., BW16 and Leech24.
\end{remark}


The message space of $m$ is $\Lambda_2/\Lambda_3$, which can be represented by  $N$ bits. The value of $N$ is given by
\begin{align}
N&=\log_2(\mathsf{Vol}(\Lambda_3)/\mathsf{Vol}(\Lambda_2))=n/\ell {\textstyle\sum\nolimits}_{i=1}^{\ell}\log _{2}(\pi_{\mathsf{LCM}}/\pi_{i}), 
\end{align}%
where $\pi_{\mathsf{LCM}}$ and $\pi_{i}$ are defined in Lemma 1. In other words, the equivalent message space of $m$ is $\mathcal{M}_{2,N} = \{0,1\}^{N}$. Consequently, the CER of KRM-$\Lambda$ can be computed by
\begin{equation}
\setlength{\abovedisplayskip}{3pt}
\setlength{\belowdisplayskip}{3pt}
\text{CER}=\dfrac{knd_u+nd_v}{N}, 
\end{equation}%
where $d_v=p-t$ and the values of $t$ are given in Example 1.

\begin{example}
One can easily calculate ${N}=256$ and $320$ bits for $\Lambda= \mathsf{E_8}$ and $\mathsf{BW16}$, respectively. We see that a higher-dimensional $\Lambda$ allows exchanging a larger secret.
\end{example}

\vspace{-1mm}

\subsection{Security}
Along the same line in \cite{MLWEE82021},  we can use Lemma 2 to show the proposed KRM-$\Lambda$ is IND-CPA secure. Similarly to Kyber, if the DFR is small, we can obtain an IND-CCA secure  KRM-$\Lambda$ using th Fujisaki-Okamoto transform \cite{modFOT2017} applied to the IND-CPA secure KRM-$\Lambda$. In the next section, we will study the DFR of KRM-$\Lambda$.

\vspace{-1mm}
\section{The Analysis on DFR} 
In this section, we propose an explicit upper bound on the DFR of KRM-$\Lambda$. We show how the choice of $\Lambda$ affects DFR.
\subsection{KRM-$\Lambda$ Decoding Noise}
A sufficient condition for correct decryption is \cite{MLWEE82021}:
\begin{align}
\mathsf{Q}_{\Lambda_2}(\boldsymbol{\mathrm{t}}^{T}\boldsymbol{\mathrm{r}}+e_2-\bs^{T}\bu' + c_v) &=0, \label{dec_con}
\end{align}
where ${c}_v  \leftarrow \mathcal{U}(\mathcal{V}_0(\Lambda_1))$ is the quantization noise, and 
\begin{equation}
\bu'=\mathsf{Decompress}_{q}(\bu,d_{u}),
\end{equation}
is given in Algorithm 3.3.3. According to \cite{Kyber2021}, we have
\begin{equation}
\bu'=\buone + \mathbf{c}_u, \label{u_prime}
\end{equation}
where $\mathbf{c}_{u}\leftarrow \psi
_{d_{u}}^{k}$ is the rounding noise generated due to the compression operation in Algorithm 3.3.2. According to (\ref{dec_con}) and (\ref{u_prime}), the KRM-$\Lambda$ decoding noise is given by
\begin{align}
n_e &= \boldsymbol{\mathrm{t}}^{T}\boldsymbol{\mathrm{r}}+e_2-\bs^{T}\bu' + c_v \notag \\
&= \mathbf{e}^{T}\mathbf{r}+e_{2}+{c}_{v}-\mathbf{s}^{T}\left( \mathbf{e}%
_{1}+\mathbf{c}_{u}\right).  \label{q_noise_model}
\end{align}

Similar to Kyber, the elements in $c_v$ or $\mathbf{c}_u$ are assumed to be i.i.d. and independent of other terms in (\ref{q_noise_model}).

\begin{theo}
According to the Central Limit Theorem (CLT), the distribution of $n_e$ asymptotically approaches the sum of multivariate normal and uniform random variables:
{\setlength{\abovedisplayskip}{3pt} \setlength{\belowdisplayskip}{3pt}
\begin{equation}
n_{e}\leftarrow \mathcal{N}(0,\sigma _{G}^{2}I_{n}) +\mathcal{U}%
(\mathcal{V}_0(\Lambda_1)),  \label{L_D_ne}
\end{equation}%
}where $\sigma _{G}^{2}=kn\eta _{{1}}^2/4+ kn\eta_{_1}/2 \cdot (\eta_2/2+\var(\psi_{d_u}))+\eta_2/2$. The values of $\var(\psi_{d_u})$ are listed in Table \ref{Du_Kyber_Var}.
\end{theo}
\begin{proof}
Proof is similar to Theorem 1 in \cite{liu2023lattice}, thus omitted.
\end{proof}

\vspace{-5mm}
\begin{table}[ht]
\caption{Values of $\var(\psi_{d_u})$ with Kyber Modulus $q=3329$}
\label{Du_Kyber_Var}\centering
\vspace{-3mm} 
\begin{tabular}{|c|c|c|c|c|}
\hline
$d_u$ & $11$ & $10$  & $9$ \\ \hline
$\var(\psi_{d_u})$  &$0.38$ \cite{liu2023lattice} & $0.9$ \cite{Kyber2021} & $3.8$ \cite{liu2023lattice} \\ \hline
\end{tabular}
\vspace{-3mm}
\end{table}

\subsection{Deriving the DFR} 
Combining (\ref{dec_con}) and (\ref{q_noise_model}), The DFR of KRM-$\Lambda$ is given by
\begin{equation}
\setlength{\abovedisplayskip}{3pt}
\setlength{\belowdisplayskip}{3pt}
\delta = \Pr(n_e \notin \mathcal{V}_0(\Lambda_2)).
\end{equation}
Since $\Lambda_2 = 2^{{p} -{t}}\lfloor q/2^{{p}}\rfloor (\Lambda)^{n/\ell}$, according to Theorem 1, we can divide $n_e$ into $n/\ell$ blocks, which are independently and identically distributed. Without loss of generality, we let $n_{e}^{(\ell)}:=[n_{e,1},n_{e,2},\ldots,n_{e,\ell}]^{T}$
to be the $\ell$ coefficients in one block. We can rewrite $\delta$ as
\begin{align}
\delta &= 1-\Pr(n_e^{(\ell)} \in \mathcal{V}_0(2^{{p} -{t}}\lfloor q/2^{{p}}\rfloor\Lambda))^{n/\ell} \notag \\
& \leq 1- \Pr(\Vert n_e^{(\ell)} \Vert \leq \mathsf{r_{pack}}(2^{{p} -{t}}\lfloor q/2^{{p}}\rfloor\Lambda))^{n/\ell}, \label{delta_UB}
\end{align}
where $\mathsf{r_{pack}}(\Lambda)$ is the packing radius of a lattice $\Lambda$.

\begin{lemma}
$\Pr \left( \Vert {n}_{e}^{(\ell)}\Vert \leq z\right) \geq 1 - Q_{\ell/2}\left(
\frac{\mathsf{r_{cov}}(\lfloor q/2^{{p}}\rfloor\Lambda)}{\sigma _{G}},\frac{z}{%
\sigma _{G}}\right)$, for $z \geq 0$, where $%
Q_{M}\left(a,b\right) $ is the generalised Marcum Q-function, and $\mathsf{r_{cov}}(\Lambda)$ is the covering radius of a lattice $\Lambda$.
\end{lemma}

\begin{proof}
Given Theorem 1, we can write
$n_{e,i}=x_{i}+y_{i}$, where $x_{i}\leftarrow \mathcal{N}( 0,\sigma
_{G}^{2}) $ and $\mathbf{y}=[y_{1},\ldots, y_{\ell}]^T \leftarrow \mathcal{U}(\mathcal{V}_0(\lfloor q/2^{{p}}\rfloor\Lambda))$. Let $N_U$ be the sample space size of $\mathbf{y}$.
\begin{eqnarray*}
&&\Pr \left( \Vert n_{e}^{(\ell)}\Vert \leq z\right)  \\
&=&\Pr \left( \sqrt{{\textstyle\sum\nolimits}_{i=1}^{\ell}\left( x_{i}+y_{i}\right) ^{2}}%
\leq z\right)  \\
&=&{\textstyle\sum\nolimits}_{j=1}^{N_{U}}\Pr \left( \sqrt{{\textstyle\sum\nolimits}_{i=1}^{\ell}%
\left( x_{i}+\mu _{j,i}\right) ^{2}}\leq z\left\vert \mathbf{y}=\mathbf{\mu }%
_{j}\right. \right) \Pr \left( \mathbf{y}=\mathbf{\mu }_{j}\right)  \\
&=&\frac{1}{N_{U}}{\textstyle\sum\nolimits}_{j=1}^{N_{U}}\Pr \left( \sqrt{%
{\textstyle\sum\nolimits}_{i=1}^{\ell}\left( x_{i}+\mu _{j,i}\right) ^{2}}\leq z\left\vert
\mathbf{y}=\mathbf{\mu }_{j}\right. \right),
\end{eqnarray*}
where $\mathbf{\mu }_{j}=[\mu _{j,1},\mu _{j,2},\ldots ,\mu _{j,\ell}]^{T}$ is a
sample point in the sample space of $\mathbf{y}$.
Given $y_{i}=\mu _{j,i}$, $\sqrt{{\textstyle\sum\nolimits}_{i=1}^{\ell}\left( x_{i}+\mu
_{j,i}\right) ^{2}}$ follows non-central chi distribution, i.e.,%
\begin{align*}
&\Pr \left( \sqrt{{{\textstyle\sum\nolimits}_{i=1}^{\ell}\left( x_{i}+\mu_{j,i}\right) ^{2}%
}/{\sigma _{G}^{2}}}\leq {z}/{\sigma _{G}}\left\vert \mathbf{y}_{j}=%
\mathbf{\mu }_{j}\right. \right)  \notag \\
&=1-Q_{\ell/2}\left( \sqrt{{%
{\textstyle\sum\nolimits}_{i=1}^{\ell}\mu _{j,i}^{2}}/{\sigma _{G}^{2}}},{z}/{\sigma
_{G}}\right),
\end{align*}%
where $Q_{M}\left( a,b\right) $ is the generalized Marcum Q-function. Since $%
Q_{M}\left( a,b\right) $ is strictly increasing in $a$ for all $a\geqslant 0$%
, we have%
\begin{align*}
\Pr \left( \Vert n_{e}^{(l)}\Vert \leq z\right)  &= 1-\frac{1}{N_{U}}%
{\textstyle\sum\nolimits}_{j=1}^{N_{U}}
Q_{\ell/2}\left( \sqrt{{%
{\textstyle\sum\nolimits}_{i=1}^{\ell}\mu _{j,i}^{2}}/{\sigma _{G}^{2}}},{z}/{\sigma
_{G}}\right)  \\
&\geqslant 1-Q_{\ell/2}\left( {\mathsf{r_{cov}}(\lfloor q/2^{{p}}\rfloor\Lambda)/\sigma_{G} 
},{z}/{\sigma _{G}}\right).
\end{align*}
\end{proof}
Given Lemma 3 and (\ref{delta_UB}), we have the following theorem.
\begin{theo}
The DFR of the proposed KRM-$\Lambda$ scheme is upper bounded by
\begin{align}~\label{eq:dfrpackcov}
\delta & \leq 1-\left(1 - Q_{\ell/2}\left(
\frac{\mathsf{r_{cov}}(\lfloor q/2^{{p}}\rfloor\Lambda)}{\sigma _{G}},\frac{\mathsf{r_{pack}}(2^{{p} -{t}}\lfloor q/2^{{p}}\rfloor\Lambda)}{%
\sigma_{G}}\right)\right)^{n/\ell}, 
\end{align}
where $\sigma_{G}$ is given in Theorem 1, and $t$ is given in (\ref{value_t}). The values of $(\mathsf{r_{pack}}, \mathsf{r_\mathsf{cov}})$ are given in Table \ref{latttice_para}.
\end{theo}


\begin{table}[th]
\centering
\begin{threeparttable}[b]
\caption{\mbox{$(\mathsf{r_{pack}}, \mathsf{r_\mathsf{cov}})$ for different lattice $\Lambda$ \cite{BK:Conway93}}}
\label{latttice_para}
\vspace{-3mm} 
\begin{tabular}{|c|c|c|c|}
\hline
& $\mathsf{E_8}$ & $\mathsf{BW16} \tnote{1}$  & $\mathsf{Leech24}$\tnote{1}\\ \hline
$\mathsf{r_{pack}}$  &$\sqrt{2}/2$ & $\sqrt{2}$ & $2\sqrt{2}$  \\ \hline
$\mathsf{r_{cov}}$  &$1$ & $\sqrt{6}$ & $4$ \\ \hline
\end{tabular}
\begin{tablenotes}
       \item [1] We scale the original generator matrix to an integer matrix $\mathbf{B} \in \mathbb{Z}^{\ell \times \ell}$.
\end{tablenotes}
\end{threeparttable}
\vspace{-5mm}
\end{table}

\begin{remark}
Theorem 2 shows how the security parameters $(q,k, \eta_1, \eta_2)$ and compression parameters $(d_u, d_v=p-t)$ affect the DFR of a KRM-$\Lambda$ scheme. Since  $Q_{M}\left(a,b\right) $ is strictly increasing in $a$ and is strictly decreasing in $b$, it also tells us to select $\Lambda$ with large $\mathsf{r_{pack}}$ and small $\mathsf{r_{cov}}$.
\end{remark}

\vspace{-1mm}
\section{KRM-$\Lambda$ with Parameters used in Kyber }
In this section, we give an example of KRM-$\Lambda$, using the same security parameters $(q=3329, k=3,\eta_1=2,\eta_2=2)$ as in KYBER768.
We tune the values of the compression parameters $(d_u, d_v)$, to obtain a good trade-off between CER and DFR. Note that $\hat{q}=3328$, $p=5$, and $d_v=p-t$, where the values of $t$ are given in Example 1. For comparison purposes, we define the CER reduction ratio:
\begin{equation}
\setlength{\abovedisplayskip}{3pt}
\setlength{\belowdisplayskip}{3pt}
\text{CER-R}\triangleq 1-\dfrac{\text{CER of } \text{KRM}-\Lambda}{\text{CER of KYBER768}}. 
\end{equation}

\vspace{-0mm}
\begin{table}[th]
\centering
\begin{threeparttable}[b]
\caption{KRM-$\Lambda$ vs. KYBER768: $(q=3329, k=3,\eta_1=2,\eta_2=2)$}
\label{lattice_quantizer}\centering
\vspace{-3mm} 
\begin{tabular}{|c|c|c|c|c|}
\hline
& KYBER768 & KRM-$\mathsf{E_8}$ & KRM-$\mathsf{BW16}$& KRM-$\mathsf{Leech24}$\tnote{1} \\ \hline
$d_u$ & $10$ & $9$ & $10$ &$10$\\ \hline
$d_v$ & $4$ & $4$ & $3$ &$2$\\ \hline
$N$ & $256$ & $256$ & $320$ &$380$ \\ \hline
CER & $34$ & $31$ & $26.4$ & $21.6$\\ \hline
CER-R & $0\%$ & $8.82\%$ & $22.35\%$ & $36.47\%$ \\ \hline
$\delta$ & $2^{-164}$ & $2^{-174}$ & $2^{-263}$ & $2^{-172}$ \\ \hline
\end{tabular}
\begin{tablenotes}
       \item [1] Since $256=10\times 24+16$, for Leech24, we consider $\Lambda_1= \lfloor q/2^{{p}}\rfloor(\text{Leech24})^{10}\times\text{BW16}$, i.e., $10$ Leech24 quantization codewords with $d_v=2$ and $1$ BW16 quantization codeword with $d_v=3$. 
\end{tablenotes}
\vspace{-1mm}
\end{threeparttable}
\end{table}
The results are given in Table \ref{lattice_quantizer}. Compared to Table \ref{Kyber_Par}, we observe that the proposed KRM-$\Lambda$ outperforms the original Kyber \cite{NISTpqcdraft2023} and the original KRM-$\mathsf{E_8}$ \cite{MLWEE82021}, in terms of CER and DFR. Since the security parameters are the same as in Kyber, the security arguments remain the same (see Table \ref{Kyber_Security}).

\begin{remark}
Many applications prefer a fixed plaintext size, e.g., $256$ bits. A simple trick is to shorten $v$. Let $\kappa = \left\lceil 256\ell / N \right\rceil$ and $(\Lambda_1 =\lfloor q/2^{{p}}\rfloor \Lambda^{\kappa},  \Lambda_2 = 2^{{p} -{t}}\Lambda_1,  \Lambda_3 =2^{{p}} \lfloor q/2^{{p}} \rfloor \mathbb{Z}^{\kappa\ell})$. We only quantize the first $\kappa \ell$ coefficients in $\boldsymbol{\mathrm{t}}^{T}\boldsymbol{\mathrm{r}}+e_2$, resulting in a shortened $v$. The idea is to use a smaller lattice quantization codebook of size $\kappa N/ \ell \approx 256$ bits. The CER is computed by 
\begin{equation}
\setlength{\abovedisplayskip}{3pt}
\setlength{\belowdisplayskip}{3pt}
\text{CER}=\dfrac{knd_u+\kappa \ell d_v}{256}. 
\end{equation}%
For BW16 with $(d_u=10, d_v=3, N=320)$ in Table \ref{lattice_quantizer}, we have $\text{CER}=32.4$, which is still smaller than KYBER768.

Besides shortening, we are interested in designing a lattice quantization codebook of size exactly $256$ bits. We require $\Lambda$ to have an integer sublattice $2^{t}\mathbb{Z}^{\ell}$ for some $t$, satisfying $\mathsf{Vol}(2^{t}\mathbb{Z}^{\ell})/\mathsf{Vol}(\Lambda)=2^{\ell}$. Therefore, the codebook size is $n/\ell \cdot \ell =256$ bits. At the moment, we are only aware of $\mathsf{E_8}$ that satisfies the above conditions for $(\ell=8, t=1)$. In general, we expect such $\Lambda$ to have a low density but large minimum distance~\cite{1LDPC}.
\end{remark}

\section{Conclusion}
We have proposed a framework that reduces the design of M-LWE-based key exchange protocol to a handful of lattice quantizer choices. We have also proved bounds on DFR against the common situation
of building lattices for quantization. We show that lattice quantizer is more effective than lattice encoding, in terms of reducing CER and DFR.



\vspace{0mm}
\bibliographystyle{IEEEtran}
\bibliography{IEEEabrv,LIUBIB}

\begin{thebibliography}{10}
\providecommand{\url}[1]{#1}
\csname url@samestyle\endcsname
\providecommand{\newblock}{\relax}
\providecommand{\bibinfo}[2]{#2}
\providecommand{\BIBentrySTDinterwordspacing}{\spaceskip=0pt\relax}
\providecommand{\BIBentryALTinterwordstretchfactor}{4}
\providecommand{\BIBentryALTinterwordspacing}{\spaceskip=\fontdimen2\font plus
\BIBentryALTinterwordstretchfactor\fontdimen3\font minus
  \fontdimen4\font\relax}
\providecommand{\BIBforeignlanguage}[2]{{%
\expandafter\ifx\csname l@#1\endcsname\relax
\typeout{** WARNING: IEEEtran.bst: No hyphenation pattern has been}%
\typeout{** loaded for the language `#1'. Using the pattern for}%
\typeout{** the default language instead.}%
\else
\language=\csname l@#1\endcsname
\fi
#2}}
\providecommand{\BIBdecl}{\relax}
\BIBdecl

\bibitem{NISTpqcdraft2023}
{National Institute of Standards and Technology}, ``{Module-Lattice-based Key
  Encapsulation Mechanism Standard},'' \emph{Federal Information Processing
  Standards Publication (FIPS) NIST FIPS 203 ipd.}, 2023.

\bibitem{Regev05}
O.~Regev, ``On lattices, learning with errors, random linear codes, and
  cryptography,'' in \emph{Proc. {ACM} Symp. Theory Comput. ({STOC})}, 2005,
  pp. 84--93.

\bibitem{RLWE2010}
V.~Lyubashevsky, C.~Peikert, and O.~Regev, ``{On Ideal Lattices and Learning
  with Errors over Rings},'' in \emph{Advances in Cryptology -- EUROCRYPT
  2010}, H.~Gilbert, Ed.\hskip 1em plus 0.5em minus 0.4em\relax Berlin,
  Heidelberg: Springer Berlin Heidelberg, 2010, pp. 1--23.

\bibitem{RSSS17}
M.~Ro{\c{s}}ca, A.~Sakzad, D.~Stehl{\'e}, and R.~Steinfeld, ``{Middle-Product
  Learning with Errors},'' in \emph{Advances in Cryptology -- CRYPTO 2017},
  J.~Katz and H.~Shacham, Eds.\hskip 1em plus 0.5em minus 0.4em\relax Cham:
  Springer International Publishing, 2017, pp. 283--297.

\bibitem{BS+15}
\BIBentryALTinterwordspacing
A.~Langlois and D.~Stehl{\'e}, ``{Worst-case to average-case reductions for
  module lattices},'' in \emph{Des. Codes Cryptogr.}, vol.~75, 2015, pp.
  565--599. [Online]. Available:
  \url{https://doi.org/10.1007/s10623-014-9938-4}
\BIBentrySTDinterwordspacing

\bibitem{homomorphiLWE2011}
Z.~Brakerski and V.~Vaikuntanathan, ``{Efficient Fully Homomorphic Encryption
  from (Standard) LWE},'' in \emph{2011 IEEE 52nd Annual Symposium on
  Foundations of Computer Science}, 2011, pp. 97--106.

\bibitem{NISTpqcdraftDS2023}
{National Institute of Standards and Technology}, ``{Module-Lattice-Based
  Digital Signature Standard},'' \emph{Federal Information Processing Standards
  Publication (FIPS) NIST FIPS 204 ipd.}, 2023.

\bibitem{CiphertextQuantization2023}
D.~Micciancio and M.~Schultz, ``{Error Correction and Ciphertext Quantization
  in Lattice Cryptography},'' in \emph{Advances in Cryptology -- CRYPTO 2023},
  H.~Handschuh and A.~Lysyanskaya, Eds.\hskip 1em plus 0.5em minus 0.4em\relax
  Cham: Springer Nature Switzerland, 2023, pp. 648--681.

\bibitem{NewhopeECC2018}
T.~Fritzmann, T.~P{\"o}ppelmann, and J.~Sepulveda, ``Analysis of
  error-correcting codes for lattice-based key exchange,'' in \emph{Selected
  Areas in Cryptography -- SAC 2018}, C.~Cid and M.~J. Jacobson~Jr., Eds.\hskip
  1em plus 0.5em minus 0.4em\relax Cham: Springer International Publishing,
  2019, pp. 369--390.

\bibitem{FrodoCong2022}
\BIBentryALTinterwordspacing
S.~Lyu, L.~Liu, C.~Ling, J.~Lai, and H.~Chen, ``{Lattice Codes for
  Lattice-Based PKE},'' in \emph{Des. Codes Cryptogr.}, 2023. [Online].
  Available: \url{https://doi.org/10.1007/s10623-023-01321-6}
\BIBentrySTDinterwordspacing

\bibitem{liu2023lattice}
\BIBentryALTinterwordspacing
S.~Liu and A.~Sakzad, ``{Lattice Codes for CRYSTALS-Kyber},'' 2023. [Online].
  Available: \url{https://arxiv.org/abs/2308.13981}
\BIBentrySTDinterwordspacing

\bibitem{KRMPeikert2014}
C.~Peikert, ``{{Lattice Cryptography for the Internet}},'' in
  \emph{Post-Quantum Cryptography}, M.~Mosca, Ed.\hskip 1em plus 0.5em minus
  0.4em\relax Cham: Springer International Publishing, 2014, pp. 197--219.

\bibitem{NewHope2016}
\BIBentryALTinterwordspacing
E.~Alkim, L.~Ducas, T.~P{\"o}ppelmann, and P.~Schwabe, ``{Post-quantum Key
  {Exchange{\textemdash}A} New Hope},'' in \emph{25th USENIX Security Symposium
  (USENIX Security 16)}.\hskip 1em plus 0.5em minus 0.4em\relax Austin, TX:
  USENIX Association, Aug. 2016, pp. 327--343. [Online]. Available:
  \url{https://www.usenix.org/conference/usenixsecurity16/technical-sessions/presentation/alkim}
\BIBentrySTDinterwordspacing

\bibitem{MLWEE82021}
C.~Saliba, L.~Luzzi, and C.~Ling, ``A reconciliation approach to key generation
  based on {Module-LWE},'' in \emph{2021 IEEE International Symposium on
  Information Theory (ISIT)}, 2021, pp. 1636--1641.

\bibitem{Kyber2021}
\BIBentryALTinterwordspacing
R.~Avanzi, J.~Bos, L.~Ducas, E.~Kiltz, T.~Lepoint, V.~Lyubashevsky, J.~Schanck,
  P.~Schwabe, G.~Seiler, and D.~Stehl\'e, ``Algorithm specifications and
  supporting documentation (version 3.02),'' Tech. rep., Submission to the NIST
  post-quantum project, 2021. [Online]. Available:
  \url{https://pq-crystals.org/kyber/resources.shtml}
\BIBentrySTDinterwordspacing

\bibitem{DFRAttack2019}
J.-P. D'Anvers, Q.~Guo, T.~Johansson, A.~Nilsson, F.~Vercauteren, and
  I.~Verbauwhede, ``{Decryption Failure Attacks on IND-CCA Secure Lattice-Based
  Schemes},'' in \emph{Public-Key Cryptography -- PKC 2019}, D.~Lin and
  K.~Sako, Eds.\hskip 1em plus 0.5em minus 0.4em\relax Cham: Springer
  International Publishing, 2019, pp. 565--598.

\bibitem{BK:Conway93}
J.~H. Conway and N.~J.~A. Sloane, \emph{Sphere Packings, Lattices, and Groups},
  3rd~ed.\hskip 1em plus 0.5em minus 0.4em\relax New York: Springer-Verlag,
  1999.

\bibitem{VD93Leech}
A.~Vardy and Y.~Be'ery, ``Maximum likelihood decoding of the {Leech} lattice,''
  \emph{IEEE Transactions on Information Theory}, vol.~39, no.~4, pp.
  1435--1444, 1993.

\bibitem{modFOT2017}
D.~Hofheinz, K.~H{\"o}velmanns, and E.~Kiltz, ``{A Modular Analysis of the
  Fujisaki-Okamoto Transformation},'' in \emph{Theory of Cryptography},
  Y.~Kalai and L.~Reyzin, Eds.\hskip 1em plus 0.5em minus 0.4em\relax Cham:
  Springer International Publishing, 2017, pp. 341--371.

\bibitem{1LDPC}
M.-R. Sadeghi and A.~Sakzad, ``On the performance of 1-level {LDPC} lattices,''
  in \emph{2013 Iran Workshop on Communication and Information Theory}, 2013,
  pp. 1--5.

\end{thebibliography}

\end{document}